\documentclass{article}
\usepackage{fullpage,amsmath,amssymb,amsthm,gastex}
\usepackage{tikz}
\def\us{\char`\_}

\def\temper{\tau}
\def\set#1{\{\,#1\,\}}
\def\abs#1{{|\,#1\,|}}
\newtheorem{theorem}{Theorem}
\newtheorem{proposition}[theorem]{Proposition}

\newtheorem{lemma}[theorem]{Lemma}
\newtheorem{corollary}[theorem]{Corollary}

\def\rtilesys#1{\begin{picture}(14,20)\gasset{AHnb=0, Nw=0, Nh=0, Nframe=no}{#1}\end{picture}}
\def\rtilesyss#1{\begin{picture}(17,20)\gasset{AHnb=0, Nw=0, Nh=0, Nframe=no}{#1}\end{picture}}

\def\rtilew#1#2#3{\drawpolygon(8.5,0)(8.5,17)(0,8.5)\node()(10.5,8.5){#1}\node()(2,13){#2}\node()(2,4){#3}}
\def\rtilewddd#1#2#3{\rtilesys{\rtilew{#1}{#2}{#3}\drawpolygon(7.5,2.5)(7.5,14.5)(1.5,8.5)}}
\def\rtilewdss#1#2#3{\rtilesys{\rtilew{#1}{#2}{#3}\drawline(7.5,1.5)(7.5,15.5)}}
\def\rtilewdsd#1#2#3{\rtilesys{\rtilew{#1}{#2}{#3}\drawline(1,9)(7.5,2.5)(7.5,15.5)}}

\def\rtilee#1#2#3{\drawpolygon(2,17)(2,0)(10.5,8.5)\node()(0,8.5){#1}\node()(8.5,4){#2}\node()(8.5,13){#3}}
\def\rtileeddd#1#2#3{\rtilesys{\rtilee{#1}{#2}{#3}\drawpolygon(3,14.5)(3,2.5)(9,8.5)}}
\def\rtileedds#1#2#3{\rtilesys{\rtilee{#1}{#2}{#3}\drawline(3,15.5)(3,2.5)(9.5,9)}}
\def\rtileedss#1#2#3{\rtilesys{\rtilee{#1}{#2}{#3}\drawline(3,15.5)(3,1.5)}}
\def\rtileedsd#1#2#3{\rtilesys{\rtilee{#1}{#2}{#3}\drawline(9.5,8)(3,14.5)(3,1.5)}}

\def\rtilen#1#2#3{\drawpolygon(0,8.5)(17,8.5)(8.5,17)\node()(8.5,6.5){#1}\node()(14,14){#2}\node()(3,14){#3}}
\def\rtilenssd#1#2#3{\rtilesyss{\rtilen{#1}{#2}{#3}\drawline(9,16)(1.5,8.5)}}
\def\rtilensds#1#2#3{\rtilesyss{\rtilen{#1}{#2}{#3}\drawline(15,9)(8,16)}}
\def\rtilensdd#1#2#3{\rtilesyss{\rtilen{#1}{#2}{#3}\drawline(15,9)(8.5,15.5)(2,9)}}

\def\rtilendds#1#2#3{\rtilesyss{\rtilen{#1}{#2}{#3}\drawline(1,9.5)(14.5,9.5)(8,16)}}
\def\rtilendss#1#2#3{\rtilesyss{\rtilen{#1}{#2}{#3}\drawline(1,9.5)(16,9.5)}}

\def\rtiles#1#2#3{\drawpolygon(17,8.5)(0,8.5)(8.5,0)\node()(8.5,10.5){#1}\node()(3,3){#2}\node()(14,3){#3}}
\def\rtilesddd#1#2#3{\rtilesyss{\rtiles{#1}{#2}{#3}\drawpolygon(14.5,7.5)(2.5,7.5)(8.5,1.5)}}
\def\rtilesdds#1#2#3{\rtilesyss{\rtiles{#1}{#2}{#3}\drawline(16,7.5)(2.5,7.5)(9,1)}}
\def\rtilesdss#1#2#3{\rtilesyss{\rtiles{#1}{#2}{#3}\drawline(16,7.5)(1,7.5)}}

\title{Triangular Self-Assembly}
\author{Lila Kari \and Shinnosuke Seki \and Zhi Xu}
\date{The University of Western Ontario, \\
Department of Computer Science, \\
Middlesex College, \\
London, Ontario, Canada N6A 5B7 \\
{\tt\{lila, sseki, zhi\us xu\}@csd.uwo.ca} \\
\medskip
\today}
\begin{document}
\maketitle

\begin{abstract}
We discuss the self-assembly system of triangular tiles instead of
square tiles, in particular right triangular tiles and equilateral
triangular tiles. We show that the triangular tile assembly system,
either deterministic or non-deterministic, has the same power to the
square tile assembly system in computation, which is Turing
universal. By providing counter-examples, we show that the
triangular tile assembly system and the square tile assembly system
are not comparable in general. More precisely, there exists square
tile assembly system $S$ such that no triangular tile assembly
system is a division of $S$ and produces the same shape; there
exists triangular tile assembly system $T$ such that no square tile
assembly system produces the same compatible shape with border
glues. We also discuss the assembly of triangles by triangular tiles
and obtain results similar to the assembly of squares, that is to
assemble a triangular of size $O(N^2)$, the minimal number of tiles
required is in $O(\log N/\log\log N)$.
\end{abstract}

\section{Introduction}
In the nature, molecules tend to interact to form more complicated
structures of crystals and supramolecules. The spontaneous
construction of particular molecular structures is one important
topic in DNA and molecular computing, which is based on the
Watson-Crick complementarity between pairs of DNA strands.
Generally, the process is composed of two steps. First, the basic
building blocks are carefully designed and constructed by synthetic
chemistry; and then the aimed large structure is assembled by
sticking basic blocks together through Watson-Crick complementarity.
In 1996, Winfree~\cite{Winfree1996} showed how the formation of
large structures from certain DNA molecules can simulate the Blocked
Cellular Automata (BCA), which is of the same computational power to
the Turing machines. In 1998, Winfree, Liu, Wenzler, and
Seeman~\cite{Winfree&Liu&Wenzler&Seeman1998} designed and produced
two-dimensional DNA crystals in their laboratory by the method of
self-assembly.

One systematic study on this topic is the self-assembly of squares.
In 1999, Adleman~\cite{Adleman1999} proposed models of
self-assembly, which are based on the theory of Wang
tiles~\cite{Wang1961}, and studied the time complexity of linear
polymerization via ``step counting''. In 2000, Rothemund and
Winfree~\cite{Rothemund&Winfree2000} showed that to
deterministically self-assemble an $N\times N$ full square, $N^2$
tiles is required for temperature $\temper=1$ and the number of
tiles for the case of fixed temperature $\temper\geq2$ is $O(\log
N)$. In 2001, Adleman, Cheng, Goel, and
Huang~\cite{Adleman&Cheng&Goel&Huang2001} showed that $\Theta(\log
N/\log\log N)$ tiles is enough for fixed temperature $\temper\geq
2$. In 2006, Kao and Schweller~\cite{Kao&Schewller2006} showed that
if the temperature $\temper$ is allowed to change systematically,
then a constant number of tiles is enough for the self-assembly of
arbitrary $N\times N$ full square with temperature sequence of
length $O(\log N)$.

One variation on the self-assembly of squares is that we study tile
of shapes other than squares that can tile a full two-dimensional
plane; and instead of considering full squares, we discuss the
self-assembly of other particular full two-dimensional region. For
tiling a full two-dimensional plane with one single shape of regular
polygons, the only possible choice of regular polygons are
equilateral triangles, squares, and hexagons. In this paper, we
discuss the self-assembly of triangles and other shapes by
triangular tiles, more specifically, of shape of equilateral
triangles and of right triangles, respectively.

In Section~2, we will introduce the definition of triangular tile
assembly system. In Section~3, we discuss the computational power of
the triangular tile assembly system, and show that it is Turing
universal. In Section~4, we compare the square tile assembly system
and triangular tile assembly system in the aspect of shape
complexity and show that the two types of system are not comparable.
In Section~5, we discuss the assembly of triangles. In the last
section, we summarize the obtained results.

\section{Definitions}
All discussion in this paper is on a two-dimensional plane. Before
we discuss the right triangular tiles and the equilateral triangular
tiles respectively, we first give a uniform definition of the
\emph{Tile Assembly Model} (TAM).

Similar to the square tiles, we define a \emph{triangular tile} to
be an triangle of particular shape (right triangle or equilateral
triangle) with each side being colored from the set $\Sigma$ of
``glues''. Without loss of generality, we assume that the shortest
side of a triangular tile is of unit length. We also assume that a
triangular tile cannot be rotated nor flipped over. Both square
tiles and triangular tiles are called tiles.

The particular non-interactive glue is denoted by $\phi$ and we
always assume $\phi\in\Sigma$. The \emph{temperature} $\temper$ is a
real number, which presents under which the assembly is proceeded,
and the set of all valid temperature is denoted by $\mathcal{W}$. A
\emph{strength function} $g:\Sigma\times\Sigma\to\mathcal{W}$ is
defined such that $g(\gamma,\gamma')=g(\gamma',\gamma)$ and
$g(\phi,\gamma)=0$. In particular, we are interested in the discrete
case $\temper\in\mathcal{N}$, $\Sigma=\Gamma\times\mathcal{N}$ and
$g((a,n),(a',n'))=n$ if $a=a',n=n'$ otherwise $g((a,n),(a',n'))=0$,
where $\mathcal{N}=0,1,\ldots$ are non-negative integers.

We say two tiles can \emph{stick} together if they can be physically
put adjacent by the sides $\gamma$ and $\gamma'$ of the same length
such that $g(\gamma,\gamma')\geq\temper$. A tile can stick to a set
of tiles if they can be physically put adjacent by the sides
$\gamma_i$ and $\gamma'_i$ of the same length such that $\sum_i
g(\gamma_i,\gamma'_i)\geq\temper$. A \emph{super-tile} is a set of
tiles that stick to each other such that no two tiles overlap and
for any two tiles there is a path of sticked edges between them.
%, any two adjacent tiles stick together, and the occupied region is
% connected.
We also call a single tile super-tile.

A \emph{tile assembly system} is a tuple $S=(T,s,g,\temper)$, where
$T$ is a finite set of tiles, $s\in T$ is a particular super-tile
called \emph{seed}, $g$ is a strength function, and $\temper$ is the
temperature. The produce of a tile assembly system is a super-tile
$st$ such that there is a super-tile sequence
$s=st_0,st_1,st_2,\ldots,st=st_n$, where $st_{i+1}$ is obtained by
stick one tile in $T$ to $st_i$ under temperature $\temper$ and no
tile in $T$ can be stick to $st_n$ to obtain a bigger super-tile. A
tile assembly system is \emph{deterministic} if its produce is
unique regardless of the different choice of super-tile at each
step.

In analogy to the self-assembly of full square, in what follows we
will study the self-assembly of full triangles and other full shapes
by right triangular tiles and equilateral triangular tiles,
respectively. Here ``full'' means the pair of common edges of every
two adjacent tiles in the produce has a positive strength.
% can stick together is a super-tile.

\emph{Right triangular tiles} are triangular tiles of the shape of
right triangles with the right angle point to four possible
directions as illustrated in Figure~\ref{fig:ttile}. More formally,
a right triangular tile is described by
$(\gamma_1,\gamma_2,\gamma_3,k)$, where $\gamma_i\in\Sigma$ are
glues on sides of the tile in the counter-clockwise order starting
from the longest side and $k\in\set{\tt e,n,w,s}$ presents the
direction of the right angle. \emph{Equilateral triangular tiles}
are triangular tiles of the shape of equilateral triangles that are
either in an upward position or in a downward position as
illustrated in Figure~\ref{fig:ttile}. More formally, it is
described by $(\gamma_1,\gamma_2,\gamma_3,k)$, where $k\in\set{\tt
u,d}$ presents the two positions and $\gamma_i\in\Sigma$ are glues
on sides of the tile in the counter-clockwise order starting from
the horizontal side.

\begin{figure}
\centering
  \begin{picture}(80,10)
  \gasset{AHnb=0, Nw=0, Nh=0, Nframe=no}
  \drawpolygon(0,17)(0,0)(8.5,8.5)\node()(0,8.5){$\gamma_1$}\node()(4.25,4.25){$\gamma_2$}\node()(4.25,12.75){$\gamma_3$}
  \drawpolygon(10,8.5)(27,8.5)(18.5,17)\node()(18.5,8.5){$\gamma_1$}\node()(22.75,12.75){$\gamma_2$}\node()(14.25,12.75){$\gamma_3$}
  \drawpolygon(38.5,0)(38.5,17)(30,8.5)\node()(38.5,8.5){$\gamma_1$}\node()(34.25,12.75){$\gamma_2$}\node()(34.25,4.25){$\gamma_3$}
  \drawpolygon(57,8.5)(40,8.5)(48.5,0)\node()(48.5,8.5){$\gamma_1$}\node()(44.25,4.25){$\gamma_2$}\node()(52.75,4.75){$\gamma_3$}

  \drawpolygon(55,0)(67,0)(61,10)\node()(61,0){$\gamma_1$}\node()(64,5){$\gamma_2$}\node()(58,5){$\gamma_3$}
  \drawpolygon(65,10)(71,0)(77,10)\node()(71,10){$\gamma_1$}\node()(68,5){$\gamma_2$}\node()(74,5){$\gamma_3$}
  \end{picture}
  \caption{Four right triangular tiles and two equilateral triangular tiles}\label{fig:ttile}
\end{figure}
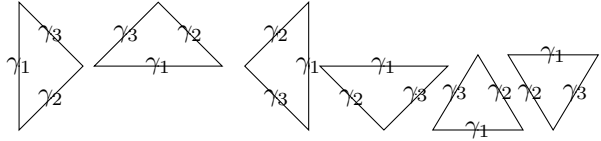

%-----------------------------------------------------------------------------
    \section{Computational Complexity}
%-----------------------------------------------------------------------------

Tiling a plane is equivalent to attaching tiles onto a lattice of a coordinate system on the plane.
The choice of coordinate system is arbitrary, but square tile systems are to choose the rectangular coordinate system $C_R$.
In contrast, the oblique coordinate system is rather natural as a pasted board of triangular tiles.
It seems reasonable to say that the oblique coordinate system $C_{\pi/3}$ whose two axes intersect with $\pi/3$ is the best choice for equilateral triangular tiles.
The right triangular tile accords with both rectangular and oblique coordinate systems because it can tessellate unlike equilateral triangular tiles.
The conversion among these coordinate systems can be done by affine transformations.

\begin{figure}
\begin{center}
\begin{tikzpicture}
    \draw (0,1.2) -- node {$N_i$} (1.2,1.2) -- node {$E_i$} (1.2,0) -- node {$S_i$} (0,0) -- node {$W_i$} (0,1.2);

    \draw[very thick,->] (2,0.6) -- (2.5,0.6);

    \draw (3,0) -- node {$S_i$} (4.2,0) -- node {$\gamma_i$} (3.6,1) -- node {$W_i$} (3,0);
    \draw (3.8,1.2) -- node {$\gamma_i$} (4.4,0.2) -- node {$E_i$} (5,1.2) -- node {$N_i$} (3.8,1.2);
\end{tikzpicture}
\end{center}
\caption{Conversion of a Wang system with square tiles into an equivalenct Wang system with equilateral triangular tiles}
\label{fig:square_equi_simulation}
\end{figure}
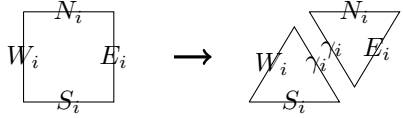

As implied in Figure~\ref{fig:square_equi_simulation}, so far as Wang tile system is concerned, whether the tile shape is square, equilateral triangle, or right triangle does not matter because the Wang tile system does not have the notion of growth by time or temperature, and imposes that any abutting edges have to have the same glue.
Several problems on the computational complexity of Wang tile system was studied by Robinson in 1971 \cite{Robinson1971}.
Among them, one important problem is the tiling full plane problem: given a Wang tile system, decide whether any product of that system is not a full plane.
The argument so far should make it clear that we can obtain analogous results for Wang system with triangular tiles.
For example, tiling full plane problem is undecidable for a given Wang system with (equilateral, right) triangular tiles.

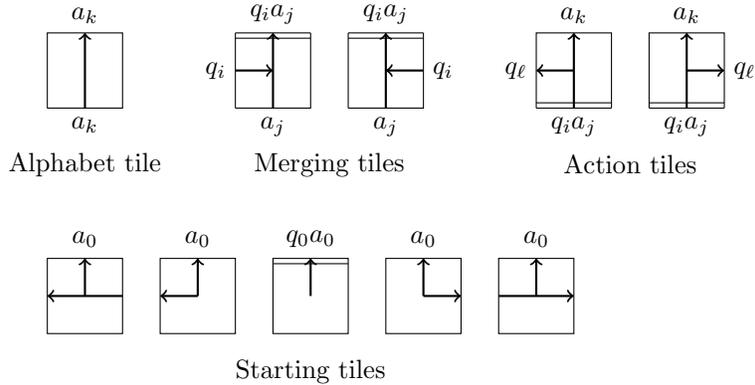
\begin{figure}
\begin{center}
\begin{tikzpicture}
    \draw (0,0) -- node[below] {$a_k$} (1,0) -- (1,1) -- node[above] {$a_k$} (0,1) -- (0,0);
    \draw [thick,->] (0.5,0) -- (0.5,1);

    \node at (0.5,-0.75) {Alphabet tile};

    \draw (2.5,0) -- node[below] {$a_j$} (3.5,0) -- (3.5,1);
    \draw (3.5,1) -- node[above] {$q_i a_j$} (2.5,1) [yshift=-2pt] (3.5,1) -- (2.5,1);
    \draw (2.5,1) -- node[left] {$q_i$} (2.5,0);
    \draw [thick,->] (3,0) -- (3,1);
    \draw [thick,->] (2.5,0.5) -- (3,0.5);

    \draw (4,0) -- node[below] {$a_j$} (5,0) -- node[right] {$q_i$} (5,1);
    \draw (5,1) -- node[above] {$q_i a_j$} (4,1) [yshift=-2pt] (5,1) -- (4,1);
    \draw (4,1) -- (4,0);
    \draw [thick,->] (4.5,0) -- (4.5,1);
    \draw [thick,->] (5,0.5) -- (4.5,0.5);

    \node at (3.75,-0.75) {Merging tiles};

    \draw (6.5,0) -- node[below] {$q_i a_j$} (7.5,0) [yshift=2pt] (6.5,0) -- (7.5,0);
    \draw (7.5,0) -- (7.5,1) -- node[above] {$a_k$} (6.5,1) -- node[left] {$q_\ell$} (6.5,0);
    \draw [thick,->] (7,0) -- (7,1);
    \draw [thick,->] (7,0.5) -- (6.5,0.5);

    \draw (8,0) -- node[below] {$q_i a_j$} (9,0) [yshift=2pt] (8,0) -- (9,0);
    \draw (9,0) -- node[right] {$q_\ell$} (9,1) -- node[above] {$a_k$} (8,1) -- (8,0);
    \draw [thick,->] (8.5,0) -- (8.5,1);
    \draw [thick,->] (8.5,0.5) -- (9,0.5);

    \node at (7.75,-0.75) {Action tiles};

    \draw (0,-2) rectangle node[above=15pt] {$a_0$} (1,-3);
    \draw (1.5,-2) rectangle node[above=15pt] {$a_0$} (2.5,-3);
    \draw (3,-2) rectangle node[above=15pt] {$q_0 a_0$} (4,-3);
    \draw (4.5,-2) rectangle node[above=15pt] {$a_0$} (5.5,-3);
    \draw (6,-2) rectangle node[above=15pt] {$a_0$} (7,-3);
    \draw [yshift=-2pt] (3,-2) -- (4,-2);

    \draw [thick,->] (1,-2.5) -- (0,-2.5);
    \draw [thick,->] (0.5,-2.5) -- (0.5,-2);
    \draw [thick,->] (2,-2.5) -- (1.5,-2.5);
    \draw [thick,->] (2,-2.5) -- (2,-2);
    \draw [thick,->] (3.5,-2.5) -- (3.5,-2);
    \draw [thick,->] (5,-2.5) -- (5.5,-2.5);
    \draw [thick,->] (5,-2.5) -- (5,-2);
    \draw [thick,->] (6,-2.5) -- (7,-2.5);
    \draw [thick,->] (6.5,-2.5) -- (6.5,-2);

    \node at (3.5,-3.5) {Starting tiles};
\end{tikzpicture}
\end{center}
\caption{
    Tiles of a square TAM which simulates a given Turing Machine at temperature $\tau = 2$.
    A tile can stick to its neighbor via a single-lined edge with a glue strength 1 or via double-lined one with strength 2, but the glue works only when the abutting edges share the same label and the directions of their arrow heads (if any) must match (head with tail).
}
\label{fig:square_TM}
\end{figure}

%In 1971, Robinson~\cite{Robinson1971} studied several problems on the computational complexity of square tile assembly system.
%In this section, we will discuss the analogues on the triangular tile assembly system.
%The following problem is called tiling full plane in the paper: given a triangular tile assembly system, decide whether any produce of that system is not a full plane.

This conversion may still work for non-deterministic tile assembly models, but it does not work any more for deterministic ones.
Let us verify this statement by trying to simulate a Turing machine by the triangular TAM thus obtained.
Based on the conversion, the tile in Figure~\ref{fig:square_TM} which merges the state $q_i$ from the right to the letter $a_j$ is split into $(a_j, \gamma, \phi, \mathtt{u})$ and $(q_i a_j, \gamma, q_i, \mathtt{d})$.
What is important is that two inputs of the square merging tile $a_j$, $q_i$ are now separated onto the two triangular tiles, and cannot cooperate until one of the tiles is stuck to the super-tile.

This failure means that the conversion requires some modification for the deterministic triangular TAM construction.
In the following, we will prove that the triangular tile assembly system is Turning universal in the sense that the tiling full plane problem can simulate the halting problem.
Throughout these proofs, it will be elucidated what modification to be required.

\begin{theorem}\label{thm:equitria_TM}
    The deterministic equilateral triangular tile assembly system is Turing universal at temperature $\tau \ge 2$.
\end{theorem}
\begin{proof}
    We simulate a given deterministic Turing machine $M = (Q, \Sigma, \Gamma, \delta, q_0, \mathtt{B}, F)$ by a deterministic equilateral triangular TAM whose tile set is shown in Figure~\ref{fig:detETTAM_TM}.
    Without loss of generality, we can assume that $M$ always moves its head when it transits.

    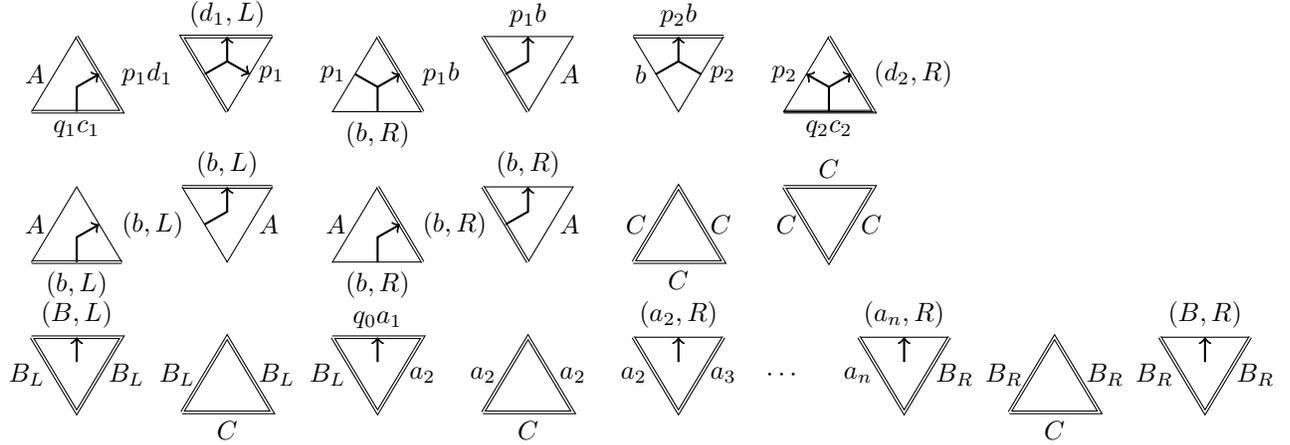

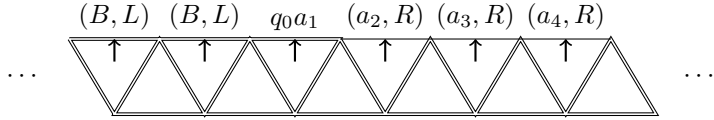
\begin{figure}
    \begin{tikzpicture}

    \path (0,-1) coordinate (BL1R);
    \path (BL1R)+(1.2,0) coordinate (BL2R);
    \path (BL1R)+(0.6,-1) coordinate (BL3R);
    \path (BL1R)+(0.6,-0.33) coordinate (BLcR);
    \path (BL1R)+(0.6,0) coordinate (BL12R);
    \path (BL1R)+(0.9,-0.5) coordinate (BL23R);
    \path (BL1R)+(0.3,-0.5) coordinate (BL31R);

    \draw [double] (BL1R) -- node[above] {$(B,L)$} (BL2R) -- node[right] {$B_L$} (BL3R) -- node[left] {$B_L$} (BL1R);
    \draw [thick,->] (BLcR) -- (BL12R);

    \path (2,-2) coordinate (BL1);
    \path (BL1)+(1.2,0) coordinate (BL2);
    \path (BL1)+(0.6,1) coordinate (BL3);
    \path (BL1)+(0.6,0.33) coordinate (BLc);
    \path (BL1)+(0.6,0) coordinate (BL12);
    \path (BL1)+(0.9,0.5) coordinate (BL23);
    \path (BL1)+(0.3,0.5) coordinate (BL31);

    \draw [double] (BL1) -- node[below] {$C$} (BL2) -- node[right] {$B_L$} (BL3) -- node[left] {$B_L$} (BL1);

    \path (4,-1) coordinate (Seed1R);
    \path (Seed1R)+(1.2,0) coordinate (Seed2R);
    \path (Seed1R)+(0.6,-1) coordinate (Seed3R);
    \path (Seed1R)+(0.6,-0.33) coordinate (SeedcR);
    \path (Seed1R)+(0.6,0) coordinate (Seed12R);
    \path (Seed1R)+(0.9,-0.5) coordinate (Seed23R);
    \path (Seed1R)+(0.3,-0.5) coordinate (Seed31R);

    \draw [double] (Seed1R) -- node[above] {$q_0 a_1$} (Seed2R) -- node[right] {$a_2$} (Seed3R) -- node[left] {$B_L$} (Seed1R);
    \draw [thick,->] (SeedcR) -- (Seed12R);

    \path (6,-2) coordinate (Sa21);
    \path (Sa21)+(1.2,0) coordinate (Sa22);
    \path (Sa21)+(0.6,1) coordinate (Sa23);
    \path (Sa21)+(0.6,0.33) coordinate (Sa2c);
    \path (Sa21)+(0.6,0) coordinate (Sa212);
    \path (Sa21)+(0.9,0.5) coordinate (Sa223);
    \path (Sa21)+(0.3,0.5) coordinate (Sa231);

    \draw [double] (Sa21) -- node[below] {$C$} (Sa22) -- node[right] {$a_2$} (Sa23) -- node[left] {$a_2$} (Sa21);

    \path (8,-1) coordinate (Sa21R);
    \path (Sa21R)+(1.2,0) coordinate (Sa22R);
    \path (Sa21R)+(0.6,-1) coordinate (Sa23R);
    \path (Sa21R)+(0.6,-0.33) coordinate (Sa2cR);
    \path (Sa21R)+(0.6,0) coordinate (Sa212R);
    \path (Sa21R)+(0.9,-0.5) coordinate (Sa223R);
    \path (Sa21R)+(0.3,-0.5) coordinate (Sa231R);

    \draw (Sa21R) -- node[above] {$(a_2,R)$} (Sa22R); \draw [double] (Sa22R) -- node[right] {$a_3$} (Sa23R) -- node[left] {$a_2$} (Sa21R);
    \draw [thick,->] (Sa2cR) -- (Sa212R);

    \node at (10,-1.5) {$\cdots$};

    \path (11,-1) coordinate (San1R);
    \path (San1R)+(1.2,0) coordinate (San2R);
    \path (San1R)+(0.6,-1) coordinate (San3R);
    \path (San1R)+(0.6,-0.33) coordinate (SancR);
    \path (San1R)+(0.6,0) coordinate (San12R);
    \path (San1R)+(0.9,-0.5) coordinate (San23R);
    \path (San1R)+(0.3,-0.5) coordinate (San31R);

    \draw (San1R) -- node[above] {$(a_n,R)$} (San2R); \draw [double] (San2R) -- node[right] {$B_R$} (San3R) -- node[left] {$a_n$} (San1R);
    \draw [thick,->] (SancR) -- (San12R);

    \path (13,-2) coordinate (BR1);
    \path (BR1)+(1.2,0) coordinate (BR2);
    \path (BR1)+(0.6,1) coordinate (BR3);
    \path (BR1)+(0.6,0.33) coordinate (BRc);
    \path (BR1)+(0.6,0) coordinate (BR12);
    \path (BR1)+(0.9,0.5) coordinate (BR23);
    \path (BR1)+(0.3,0.5) coordinate (BR31);

    \draw [double] (BR1) -- node[below] {$C$} (BR2) -- node[right] {$B_R$} (BR3) -- node[left] {$B_R$} (BR1);

    \path (15,-1) coordinate (BR1R);
    \path (BR1R)+(1.2,0) coordinate (BR2R);
    \path (BR1R)+(0.6,-1) coordinate (BR3R);
    \path (BR1R)+(0.6,-0.33) coordinate (BRcR);
    \path (BR1R)+(0.6,0) coordinate (BR12R);
    \path (BR1R)+(0.9,-0.5) coordinate (BR23R);
    \path (BR1R)+(0.3,-0.5) coordinate (BR31R);

    \draw (BR1R) -- node[above] {$(B,R)$} (BR2R); \draw [double] (BR2R) -- node[right] {$B_R$} (BR3R) -- node[left] {$B_R$} (BR1R);
    \draw [thick,->] (BRcR) -- (BR12R);

    \path (0,2) coordinate (C1);
    \path (C1)+(1.2,0) coordinate (C2);
    \path (C1)+(0.6,1) coordinate (C3);
    \path (C1)+(0.6,0.33) coordinate (Cc);
    \path (C1)+(0.6,0) coordinate (C12);
    \path (C1)+(0.9,0.5) coordinate (C23);
    \path (C1)+(0.3,0.5) coordinate (C31);

    \draw [double] (C1) -- node[below] {$q_1 c_1$} (C2) -- node[right=5pt] {$p_1 d_1$} (C3);
    \draw (C3) -- node[left] {$A$} (C1);
    \draw [thick,->] (C12) -- (Cc) -- (C23);

    \path (2,3) coordinate (C1R);
    \path (C1R)+(1.2,0) coordinate (C2R);
    \path (C1R)+(0.6,-1) coordinate (C3R);
    \path (C1R)+(0.6,-0.33) coordinate (CcR);
    \path (C1R)+(0.6,0) coordinate (C12R);
    \path (C1R)+(0.9,-0.5) coordinate (C23R);
    \path (C1R)+(0.3,-0.5) coordinate (C31R);

    \draw [double] (C3R) -- (C1R) -- node[above] {$(d_1,L)$} (C2R);
    \draw (C2R) -- node[right] {$p_1$} (C3R);

    \draw [thick,->] (C31R) -- (CcR) -- (C12R); \draw [thick,->] (CcR) -- (C23R);

    \path (4,2) coordinate (D1);
    \path (D1)+(1.2,0) coordinate (D2);
    \path (D1)+(0.6,1) coordinate (D3);
    \path (D1)+(0.6,0.33) coordinate (Dc);
    \path (D1)+(0.6,0) coordinate (D12);
    \path (D1)+(0.9,0.5) coordinate (D23);
    \path (D1)+(0.3,0.5) coordinate (D31);

    \draw (D3) -- node[left] {$p_1$} (D1) -- node[below] {$(b,R)$} (D2);
    \draw [double] (D2) -- node[right=5pt] {$p_1 b$} (D3);

    \draw [thick,->] (D31) -- (Dc) -- (D23); \draw [thick] (D12) -- (Dc);

    \path (6,3) coordinate (D1R);
    \path (D1R)+(1.2,0) coordinate (D2R);
    \path (D1R)+(0.6,-1) coordinate (D3R);
    \path (D1R)+(0.6,-0.33) coordinate (DcR);
    \path (D1R)+(0.6,0) coordinate (D12R);
    \path (D1R)+(0.9,-0.5) coordinate (D23R);
    \path (D1R)+(0.3,-0.5) coordinate (D31R);

    \draw (D1R) -- node[above] {$p_1 b$} (D2R) -- node[right] {$A$} (D3R);
    \draw [double] (D3R) -- (D1R);

    \draw [thick,->] (D31R) -- (DcR) -- (D12R);

    \path (8,3) coordinate (E1R);
    \path (E1R)+(1.2,0) coordinate (E2R);
    \path (E1R)+(0.6,-1) coordinate (E3R);
    \path (E1R)+(0.6,-0.33) coordinate (EcR);
    \path (E1R)+(0.6,0) coordinate (E12R);
    \path (E1R)+(0.9,-0.5) coordinate (E23R);
    \path (E1R)+(0.3,-0.5) coordinate (E31R);

    \draw [double] (E1R) -- node[above] {$p_2 b$} (E2R);
    \draw (E2R) -- node[right] {$p_2$} (E3R) -- node[left] {$b$} (E1R);

    \draw [thick,->] (E31R) -- (EcR) -- (E12R); \draw [thick] (E23R) -- (EcR);

    \path (10,2) coordinate (E1);
    \path (E1)+(1.2,0) coordinate (E2);
    \path (E1)+(0.6,1) coordinate (E3);
    \path (E1)+(0.6,0.33) coordinate (Ec);
    \path (E1)+(0.6,0) coordinate (E12);
    \path (E1)+(0.9,0.5) coordinate (E23);
    \path (E1)+(0.3,0.5) coordinate (E31);

    \draw [double] (E1) -- node[below] {$q_2 c_2$} (E2) -- node[right=5pt] {$(d_2,R)$} (E3);
    \draw (E3) -- node[left] {$p_2$} (E1) -- (E2);

    \draw [thick,->] (E12) -- (Ec) -- (E23); \draw [thick,->] (Ec) -- (E31);

    \path (0,0) coordinate (A1);
    \path (A1)+(1.2,0) coordinate (A2);
    \path (A1)+(0.6,1) coordinate (A3);
    \path (A1)+(0.6,0.33) coordinate (Ac);
    \path (A1)+(0.6,0) coordinate (A12);
    \path (A1)+(0.9,0.5) coordinate (A23);
    \path (A1)+(0.3,0.5) coordinate (A31);

    \draw [double] (A1) -- node[below] {$(b,L)$} (A2);
    \draw (A2) -- node[right=5pt] {$(b,L)$} (A3) -- node[left] {$A$} (A1);

    \draw [thick,->] (A12) -- (Ac) -- (A23);

    \path (2,1) coordinate (A1R);
    \path (A1R)+(1.2,0) coordinate (A2R);
    \path (A1R)+(0.6,-1) coordinate (A3R);
    \path (A1R)+(0.6,-0.33) coordinate (AcR);
    \path (A1R)+(0.6,0) coordinate (A12R);
    \path (A1R)+(0.9,-0.5) coordinate (A23R);
    \path (A1R)+(0.3,-0.5) coordinate (A31R);

    \draw [double] (A1R) -- node[above] {$(b,L)$} (A2R); \draw (A2R) -- node[right] {$A$} (A3R) -- (A1R);

    \draw [thick,->] (A31R) -- (AcR) -- (A12R);

    \path (4,0) coordinate (B1);
    \path (B1)+(1.2,0) coordinate (B2);
    \path (B1)+(0.6,1) coordinate (B3);
    \path (B1)+(0.6,0.33) coordinate (Bc);
    \path (B1)+(0.6,0) coordinate (B12);
    \path (B1)+(0.9,0.5) coordinate (B23);
    \path (B1)+(0.3,0.5) coordinate (B31);

    \draw (B3) -- node[left] {$A$} (B1) -- node[below] {$(b,R)$} (B2);
    \draw [double] (B2) -- node[right=5pt] {$(b,R)$} (B3);

    \draw [thick,->] (B12) -- (Bc) -- (B23);

    \path (6,1) coordinate (B1R);
    \path (B1R)+(1.2,0) coordinate (B2R);
    \path (B1R)+(0.6,-1) coordinate (B3R);
    \path (B1R)+(0.6,-0.33) coordinate (BcR);
    \path (B1R)+(0.6,0) coordinate (B12R);
    \path (B1R)+(0.9,-0.5) coordinate (B23R);
    \path (B1R)+(0.3,-0.5) coordinate (B31R);

    \draw (B1R) -- node[above] {$(b,R)$} (B2R) -- node[right] {$A$} (B3R);
    \draw [double] (B3R) -- (B1R);

    \draw [thick,->] (B31R) -- (BcR) -- (B12R);

    \path (8,0) coordinate (Bottom1);
    \path (Bottom1)+(1.2,0) coordinate (Bottom2);
    \path (Bottom1)+(0.6,1) coordinate (Bottom3);
    \path (Bottom1)+(0.6,0.33) coordinate (Bottomc);
    \path (Bottom1)+(0.6,0) coordinate (Bottom12);
    \path (Bottom1)+(0.9,0.5) coordinate (Bottom23);
    \path (Bottom1)+(0.3,0.5) coordinate (Bottom31);

    \draw [double] (Bottom1) -- node[below] {$C$} (Bottom2) -- node[right] {$C$} (Bottom3) -- node[left] {$C$} (Bottom1);

    \path (10,1) coordinate (Bottom1R);
    \path (Bottom1R)+(1.2,0) coordinate (Bottom2R);
    \path (Bottom1R)+(0.6,-1) coordinate (Bottom3R);
    \path (Bottom1R)+(0.6,-0.33) coordinate (BottomcR);
    \path (Bottom1R)+(0.6,0) coordinate (Bottom12R);
    \path (Bottom1R)+(0.9,-0.5) coordinate (Bottom23R);
    \path (Bottom1R)+(0.3,-0.5) coordinate (Bottom31R);

    \draw [double] (Bottom1R) -- node[above] {$C$} (Bottom2R) -- node[right] {$C$} (Bottom3R) -- node[left] {$C$} (Bottom1R);

    \end{tikzpicture}
    \caption{
    Equilateral triangular tiles for the simulation of Turing machines, where $a_1 \cdots a_n$ is the input, $b \in \Sigma$, $\delta(q_1,c_1) = (p_1,d_1,\mathtt{L})$, and $\delta(q_2,c_2) = (p_2,d_2,\mathtt{R})$.
    (top) action tiles and merging tiles;
    (middle) alphabet tiles for letters to the left of TM head (indicated by $L$) and for letters to the right ($R$).
        The two tiles with the label $C$ are used to fill the third and fourth quadrants;
    (bottom) tiles for initialization with the third tile as seed;
    }
    \label{fig:detETTAM_TM}
    \end{figure}

    \begin{tikzpicture}

%   \draw [double] (-3,1) -- (-2.4,2) -- (-1.8,1) -- (-1.2,2) -- (-0.6,1);

    \draw [double] (-3,1) -- node[above] {$(B,L)$} (-1.8,1) -- node[above] {$(B,L)$} (-0.6,1) -- node[above] {$q_0 a_1$} (0.6,1) -- (0,0) -- (-0.6,1) -- (-1.2,0) -- (-1.8,1) -- (-2.4,0) -- (1.2,0) -- (0.6,1) (-2.4,0) -- (-3,1) (1.2,0) -- (1.8,1) -- (2.4,0) -- (3,1) -- (3.6,0) -- (4.2,1) -- (4.8,0) -- (1.2,0);
    \draw (0.6,1) -- node[above] {$(a_2,R)$} (1.8,1) -- node[above] {$(a_3,R)$} (3,1) -- node[above] {$(a_4,R)$} (4.2,1);

    \node at (-3.6,0.5) {$\cdots$};
    \node at (5.4,0.5) {$\cdots$};

    \draw [thick,->] (-2.4,0.7) -- (-2.4,1);
    \draw [thick,->] (-1.2,0.7) -- (-1.2,1);
    \draw [thick,->] (0,0.7) -- (0,1);
    \draw [thick,->] (1.2,0.7) -- (1.2,1);
    \draw [thick,->] (2.4,0.7) -- (2.4,1);
    \draw [thick,->] (3.6,0.7) -- (3.6,1);

    \end{tikzpicture}

    Using the tiles on the bottom row of Figure~\ref{fig:detETTAM_TM}, The initial configuration $\cdots \mathtt{B} q_0 a_1 a_2 \cdots a_n \mathtt{B} \cdots$ self-assembles from the seed $(q_0 a_1, B_L, a_2, \mathtt{d})$ in a straightforward manner as shown just above.
    Each letter is coupled either with the indicater $L$ if the letter is to the left of the head or with $R$ otherwise.
    Note that the top edges with the TM head or to the left of the head are double-lined, and hence are bound to their matching bottom edges with strength 2.
    Thus, for instance, the upward alphabet tile with $L$ at its bottom can stick to these top edges without any cooperation so long as their letters match.
    This is not the case for the edges to the right of the head because their glue strength is 1.

    Let us consider the transition $\delta(q_0, a_1) = (q_1, b_1, \mathtt{R})$ first.
    Via the edges with strength 2, the upward alphabet tiles simultaneously stick to the edges located to the left of the TM head.
    In order for them to extend further by using the corresponding downward alphabet tiles, an action tile ($(q_0 a_1, q_1 b_1, A, \mathtt{u})$ in this case) has to be stuck to the super-tile.
    The action tile changes the state $q_0$ and the letter $a_1$ according to the transition to $q_1$ and $b_1$ deterministically, and its corresponding downward tile branches the letter up and the state to the right.
    Now the merging tile $((a_2,R), q_1 a_2, q_1, \mathtt{u})$ can attach by the cooperation of left and bottom edges, and the attachment of its corresponding upward tile immediately follows.
    The letters to the right of TM head are extended one by one in this manner.

    
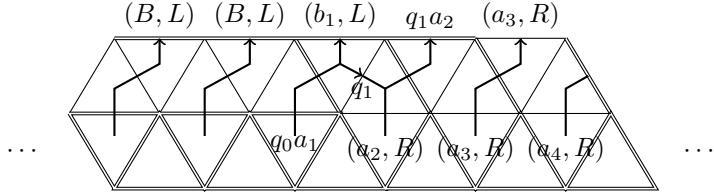
\begin{tikzpicture}

    \draw (-3,1) -- (-2.4,2) -- (-1.8,1) -- (-1.2,2) -- (-0.6,1) -- (0,2) (0.6,1) -- (1.2,2) (1.8,1) -- (2.4,2) (3,1) -- (3.6,2);
    \draw [double] (0,2) -- (0.6,1) (1.2,2) -- (1.8,1) (2.4,2) -- (3,1) (3.6,2) -- (4.2,1);

    \draw [double] (-3,1) -- (-1.8,1) -- (-0.6,1) -- node[below=5pt] {$q_0 a_1$} (0.6,1) -- (0,0) -- (-0.6,1) -- (-1.2,0) -- (-1.8,1) -- (-2.4,0) -- (1.2,0) -- (0.6,1) (-2.4,0) -- (-3,1) (1.2,0) -- (1.8,1) -- (2.4,0) -- (3,1) -- (3.6,0) -- (4.2,1) -- (4.8,0) -- (1.2,0);
    \draw (0.6,1) -- node[below=5pt] {$(a_2,R)$} (1.8,1) -- node[below=5pt] {$(a_3,R)$} (3,1) -- node[below=5pt] {$(a_4,R)$} (4.2,1);

    \draw [double] (-2.4,2) -- node[above] {$(B,L)$} (-1.2,2) -- node[above] {$(B,L)$} (0,2) -- node[above] {$(b_1,L)$} (1.2,2) -- node[above] {$q_1 a_2$} (2.4,2);
    \draw (2.4,2) -- node[above] {$(a_3,R)$} (3.6,2);

    \node at (-3.6,0.5) {$\cdots$};
    \node at (5.4,0.5) {$\cdots$};

    \draw [thick,->] (-2.4,0.7) -- (-2.4,1.33) -- (-1.8,1.66) -- (-1.8,2);
    \draw [thick,->] (-1.2,0.7) -- (-1.2,1.33) -- (-0.6,1.66) -- (-0.6,2);
    \draw [thick,->] (0,0.7) -- (0,1.33) -- (0.6,1.66) -- (0.6,2);
    \draw [thick,->] (1.2,0.7) -- (1.2,1.33) -- (1.8,1.66) -- (1.8,2);
    \draw [thick,->] (0.6,1.66) -- (0.9,1.5); \draw [thick] (0.9,1.5) -- (1.2,1.33);
    \draw [thick,->] (2.4,0.7) -- (2.4,1.33) -- (3,1.66) -- (3,2);
    \draw [thick] (3.6,0.7) -- (3.6,1.33) -- (3.9,1.5);

    \node at (0.9,1.3) {$q_1$};

    \end{tikzpicture}

    The transition $\delta(q_1, a_2) = (q_2, b_2, \mathtt{L})$ is simulated essentially in the same manner as the previous simulation so that it may suffice to illustrate it as follows:

    
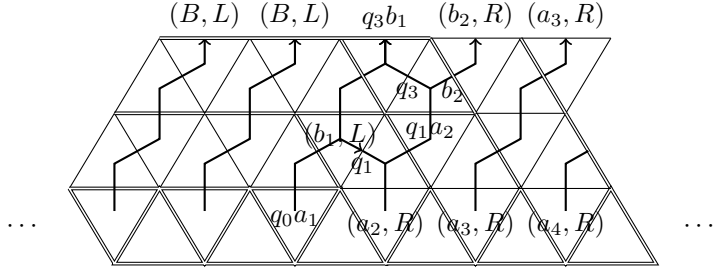
\begin{tikzpicture}

    \draw (-3,1) -- (-2.4,2) -- (-1.8,1) -- (-1.2,2) -- (-0.6,1) -- (0,2) (0.6,1) -- (1.2,2) (1.8,1) -- (2.4,2) (3,1) -- (3.6,2);
    \draw [double] (0,2) -- (0.6,1) (1.2,2) -- (1.8,1) (2.4,2) -- (3,1) (3.6,2) -- (4.2,1);

    \draw [double] (-3,1) -- (-1.8,1) -- (-0.6,1) -- node[below=5pt] {$q_0 a_1$} (0.6,1) -- (0,0) -- (-0.6,1) -- (-1.2,0) -- (-1.8,1) -- (-2.4,0) -- (1.2,0) -- (0.6,1) (-2.4,0) -- (-3,1) (1.2,0) -- (1.8,1) -- (2.4,0) -- (3,1) -- (3.6,0) -- (4.2,1) -- (4.8,0) -- (1.2,0);
    \draw (0.6,1) -- node[below=5pt] {$(a_2,R)$} (1.8,1) -- node[below=5pt] {$(a_3,R)$} (3,1) -- node[below=5pt] {$(a_4,R)$} (4.2,1);

    \draw [double] (-2.4,2) -- (-1.2,2) -- (0,2) -- node[below] {$(b_1,L)$} (1.2,2) -- node[below] {$q_1 a_2$} (2.4,2);
    \draw (2.4,2) -- (3.6,2);

    \node at (-3.6,0.5) {$\cdots$};
    \node at (5.4,0.5) {$\cdots$};

    \draw [thick,->] (1.8,2.33) -- (1.2,2.66) -- (1.2,3);

    \draw [thick,->] (-2.4,0.7) -- (-2.4,1.33) -- (-1.8,1.66) -- (-1.8,2.33) -- (-1.2,2.66) -- (-1.2,3);
    \draw [thick,->] (-1.2,0.7) -- (-1.2,1.33) -- (-0.6,1.66) -- (-0.6,2.33) -- (0,2.66) -- (0,3);
    \draw [thick,->] (0,0.7) -- (0,1.33) -- (0.6,1.66) -- (0.6,2.33) -- (1.2,2.66) -- (1.2,3);
    \draw [thick,->] (1.2,0.7) -- (1.2,1.33) -- (1.8,1.66) -- (1.8,2.33) -- (2.4,2.66) -- (2.4,3);
    \draw [thick,->] (0.6,1.66) -- (0.9,1.5); \draw [thick] (0.9,1.5) -- (1.2,1.33);
    \draw [thick,->] (2.4,0.7) -- (2.4,1.33) -- (3,1.66) -- (3,2.33) -- (3.6,2.66) -- (3.6,3);
    \draw [thick] (3.6,0.7) -- (3.6,1.33) -- (3.9,1.5);

    \node at (0.9,1.3) {$q_1$};
    \node at (2.1,2.3) {$b_2$};
    \node at (1.5,2.3) {$q_3$};

    \draw (-2.4,2) -- (-1.8,3) -- (-1.2,2) -- (-0.6,3) -- (0,2) -- (0.6,3) (1.2,2) -- (1.8,3) (2.4,2) -- (3,3) (3.6,2) -- (4.2,3);
    \draw [double] (0.6,3) -- (1.2,2) (1.8,3) -- (2.4,2) (3,3) -- (3.6,2);

    \draw [double] (-1.8,3) -- node[above] {$(B,L)$} (-0.6,3) -- node[above] {$(B,L)$} (0.6,3) -- node[above] {$q_3b_1$} (1.8,3);
    \draw (1.8,3) -- node[above] {$(b_2,R)$} (3,3) -- node[above] {$(a_3,R)$} (4.2,3);

    \end{tikzpicture}

    This Turing machine simulator consists of at most $2n+3+4|\Sigma|+m(1+3|\Sigma|)$ tiles, where $n$ is the length of the input $a_1 \cdots a_n$, and $m$ is the number of transitions defined in this Turing machine $M$.
\end{proof}

This simulation keeps tiling the plane upward until it reaches some halting configuration, i.e., the head is in a state $q$ and is on the cell with a letter $b$ such that $\delta(q, b)$ is not defined.
So the undecidability of halting problem of Turing machine leads us from this theorem to the following corollary.

\begin{corollary}
    It is undecidable whether a given deterministic equilateral triangular tile assembly system produces a super-tile other than full plane. % valid total tile.
\end{corollary}

For any equilateral triangular tile assembly system
$S=(T,s,g,\temper)$, we define the \emph{``flattened''} right
triangular assembly system $\mathcal{F}(S)=(T',f(s),g,\temper)$,
where $T'=\set{f(t),t\in T}$, $f(\gamma_1,\gamma_2,\gamma_3,{\tt
u})= (\gamma_1,\gamma_2,\gamma_3,{\tt n})$, and
$f(\gamma_1,\gamma_2,\gamma_3,{\tt d})=
(\gamma_1,\gamma_2,\gamma_3,{\tt s})$. Then $S$ produces a
super-tile that is not the full plane if and only if
$\mathcal{F}(S)$ also produces a super-tile that is not the full
plane. Then the following corollaries of Theorem~\ref{thm:equitria_TM}
hold.

\begin{corollary}
    Deterministic right triangular tile assembly system is Turing universal at temperature $\tau$ for $\tau \ge 2$.
\end{corollary}

\begin{corollary}
    It is undecidable whether a given deterministic right triangular tile assembly system produces a super-tile other than full plane.
\end{corollary}

One advantage of right triangles over equilateral ones is that right triangles can tile the square grid, and actually there are two ways to fill a square being rotated by $\pi/4$ with right triangles: east and west triangles or north and south triangles.
This fact enables us to handle more intuitively the ``input-split'' problem which the equilateral triangular TAM has already encountered.
That is, the square tile which merges a state from the left is split into half from its left-top to right-bottom, while the tile which merges from the right is cut from right-top to left-bottom.
Figure~\ref{fig:rtri_TM} illustrates the right triangular TAM designed according to this idea, which simulates a deterministic Turing machine $M=(Q,\Sigma,\Gamma,\delta,q_0,{\tt B},F)$ on an input $a_1a_2a_3\ldots a_n$.
One can verify by definition that the given system is a deterministic right triangular tile assembly system with at most $2n+4+4|\Sigma|+m(1+2\abs{\Sigma})$ tiles (slightly better than the flattened right triangular TAM), where $n$ is the length of the input of the Turing machine $M$ and $m$ is the number of transitions defined in the Turing machine $M$.
%Furthermore, the Turing machine stop on the input if and only if the produce of the given system is not the full plane.
%Therefore, deterministic right triangular tile assembly system can simulate the Turing machine $M$ and thus is universal.

\begin{figure}
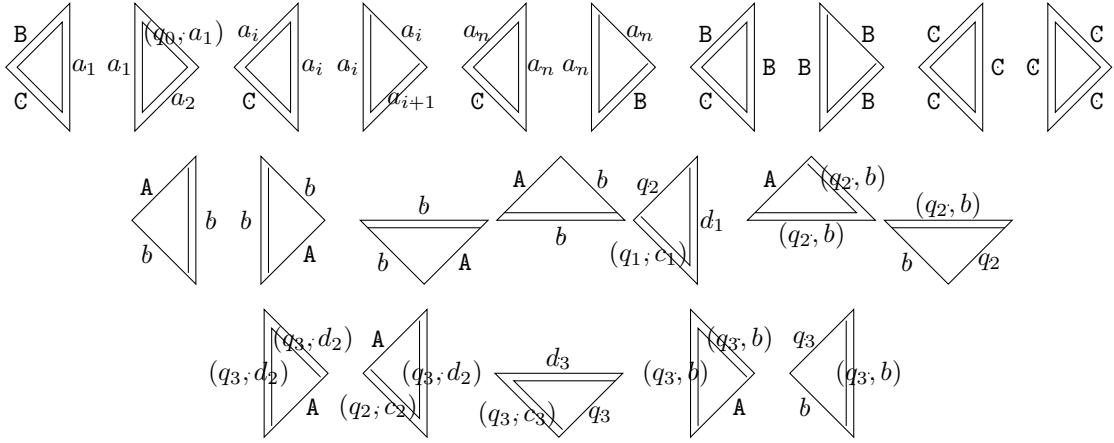

\centering
  \rtilewddd{$a_1$}{$\tt B$}{$\tt C$}
  \rtileeddd{$a_1$}{$a_2$}{$(q_0,a_1)$}
  \rtilewddd{$a_{i}$}{$a_{i}$}{$\tt C$}
  \rtileedds{$a_{i}$}{$a_{i+1}$}{$a_{i}$}
  \rtilewddd{$a_n$}{$a_n$}{$\tt C$}
  \rtileedds{$a_n$}{$\tt B$}{$a_n$}
  \rtilewddd{$\tt B$}{$\tt B$}{$\tt C$}
  \rtileedds{$\tt B$}{$\tt B$}{$\tt B$}
  \rtilewddd{$\tt C$}{$\tt C$}{$\tt C$}
  \rtileeddd{$\tt C$}{$\tt C$}{$\tt C$}

  \rtilewdss{$b$}{$\tt A$}{$b$}
  \rtileedss{$b$}{$\tt A$}{$b$}
  \rtilesdss{$b$}{$b$}{$\tt A$}
  \rtilendss{$b$}{$b$}{$\tt A$}
  \rtilewdsd{$d_1$}{$q_2$}{$(q_1,c_1)$}
  \rtilendds{$(q_2,b)$}{$(q_2,b)$}{$\tt A$}
  \rtilesdss{$(q_2,b)$}{$b$}{$q_2$}

  \rtileedsd{$(q_3,d_2)$}{$\tt A$}{$(q_3,d_2)$}
  \rtilewdsd{$(q_3,d_2)$}{$\tt A$}{$(q_2,c_2)$}\quad
  \rtilesdds{$d_3$}{$(q_3,c_3)$}{$q_3$}\qquad
  \rtileedsd{$(q_3,b)$}{$\tt A$}{$(q_3,b)$}
  \rtilewdss{$(q_3,b)$}{$q_3$}{$b$}
  \caption{Right triangular tiles for the simulation of Turing machines with the seed $(a_1,{\tt B,C,w})$.
  In this figure, $a_i$ are the input, $b\in\Sigma$ iterate all letters, $q_i\in Q\setminus F$,
  $\delta(q_1,c_1)=(p_1,d_1,{\tt L})$,
  $\delta(q_2,c_2)=(p_2,d_2,{\tt N})$, and
  $\delta(q_3,c_3)=(p_3,d_3,{\tt R})$.}\label{fig:rtri_TM}
\end{figure}

%From the proof of Theorem~\ref{fig:rtri_TM}, we see that the given deterministic right triangular tile assembly system that simulate the Turing machine can be viewed as a square tile assembly system with every square tiles being cut into half and rotated by $\pi/4$.
The simulations of Turing machines by TAMs with different shapes given in this section negates the idea that the right triangular TAM or even equilateral triangular TAM might be completely equivalent to the square TAM, topologically, in spite of the different tile shapes.
We will proceed this investigation further in the next section.

\section{Shape Complexity}
We can a shape \emph{compatible} with a given type of tile assembly
system, if the region occupied by that shape on the two dimensional
plane can be tiled geometrically by the tiles. It is obviously that
if a shape is not compatible, then no super-tile of that shape can
be produced by tile assembly systems. For example, a pie cannot be
assembled by unit square tiles nor triangular tiles. So in the
remaining discussion, we only consider the assembly of compatible
shapes.

\begin{proposition}\label{prop:t1}
Any compatible shape can be produced by a non-deterministic
triangular tile assembly system with $O(1)$ tiles, and by a
deterministic triangular tile assembly system with $A$ tiles, where
$A$ is the totally number of tiles to geometrically assemble that
shape.
\end{proposition}
\begin{proof}
First we consider the equilateral triangles in the non-deterministic
cases. Consider the set of tiles
$T=\set{(a,b,c,k):a,b,c\in\set{\phi,g},k\in\set{\tt u,d}}$. All glue
of $g$ are of strength $1$ and temperature is $\temper=1$. The seed
and assembly process is as follows: the super-tile of given shape is
assembled according to the geometrical division of the region. This
can be done since the shape is compatible. At each step, a tile
sticks to the super-tile in such a way that if the tile is
surrounded by other tiles in the completed region, then every edge
of that tile is of glue $g$; otherwise, the edge that composes the
border of that region is of empty glue $\phi$.

If we make each pair of stick sides with unique glue, then the shape
can be assembled deterministically by $O(A)$ tiles.

The right triangles case is similar.
\end{proof}

Proposition~\ref{prop:t1} can be generalized to tiles of other
shapes, such as square tiles. In what follows, we only consider
deterministic tile assembly systems.

A right triangular tile assembly system $T$ is called a
\emph{division} of a square tile assembly system $S$ if for any tile
$s$ in $S$, there is a pair of tiles $t,t'$ such that on temperature
$\temper\geq1$ tiles $t,t'$ can produce $s$ with $\pi/4$ rotation;
and for any tile $t$ in $T$, there is a tile $t'$ in $T$ and a tile
$s$ in $S$ such that on temperature $\temper\geq 1$ tiles $t,t'$ can
produce $s$ with $\pi/4$ rotation. By the definition, division of a
square tile assembly system may not be unique, and a right
triangular tile assembly system can be division of two different
square tile assembly system. The number of tiles in the two systems
satisfies the inequality
  \[2\sqrt{\#S}\leq\#T\leq4\cdot\#S,\]
where $\#$ presents the number of tiles in each system. A
equilateral triangular tile assembly system $T$ is called a division
of a square tile assembly system $S$ if the flattened tile assembly
system $\mathcal{F}(T)$ is a division of $S$.

\begin{lemma}\label{lem:countereg}
There exists a square tile assembly system $S$ such that no division
of $S$ can produce the same shape with $\pi/4$ rotation.
\end{lemma}
\begin{proof}
We presents two examples here. The produce of the two square tile
assembly systems are illustrated in Figure~\ref{fig:countereg},
where each glue is unique and the strength is illustrated by the
number of parallel edges. The temperatures are of $\temper=3$ and
$\temper=2$ respectively. The number of divisions of the square tile
assembly system is finite. One can verify that none of them produce
the same shape with the original system. \qedhere

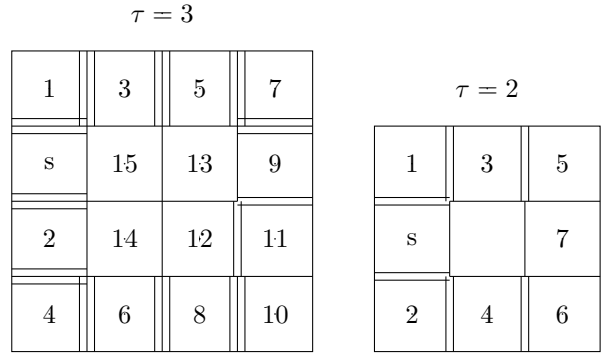
\begin{figure}
\centering
  \begin{picture}(40,40)
  \gasset{AHnb=0, Nw=0, Nh=0, Nframe=no}
  \drawpolygon(0,0)(40,0)(40,40)(0,40)
  \drawline(0,10)(40,10)\drawline(0,20)(30,20)\drawline(0,30)(40,30)
  \drawline(10,40)(10,0)\drawline(20,40)(20,0)\drawline(30,40)(30,20)\drawline(30,10)(30,0)
  \drawline(0,11)(10,11)\drawline(0,9)(10,9)\drawline(0,21)(10,21)\drawline(0,19)(10,19)
  \drawline(0,31)(10,31)\drawline(0,29)(10,29)\drawline(30,31)(40,31)\drawline(30,29)(40,29)
  \drawline(9,40)(9,30)\drawline(11,40)(11,30)\drawline(9,10)(9,0)\drawline(11,10)(11,0)
  \drawline(19,40)(19,30)\drawline(21,40)(21,30)\drawline(19,10)(19,0)\drawline(21,10)(21,0)
  \drawline(29,40)(29,30)\drawline(31,40)(31,30)\drawline(29,10)(29,0)\drawline(31,10)(31,0)
  \drawline(30,19.5)(40,19.5)\drawline(30,20.5)(40,20.5)
  \drawline(29.5,20)(29.5,10)\drawline(30.5,20)(30.5,10)
  \node()(5,25){s}\node()(20,45){$\temper=3$}\node()(5,35){1}\node()(5,15){2}
  \node()(15,35){3}\node()(5,5){4}\node()(25,35){5}\node()(15,5){6}\node()(35,35){7}\node()(25,5){8}
  \node()(35,25){9}\node()(35,5){10}\node()(35,15){11}\node()(25,15){12}
  \node()(25,25){13}\node()(15,15){14}\node()(15,25){15}
  \end{picture}
  \qquad
  \begin{picture}(30,30)
  \gasset{AHnb=0, Nw=0, Nh=0, Nframe=no}
  \drawpolygon(0,0)(30,0)(30,30)(0,30)
  \drawline(10,10)(30,10)\drawline(10,20)(30,20)\drawline(10,20)(10,10)\drawline(20,20)(20,10)
  \drawline(0,9.5)(10,9.5)\drawline(0,10.5)(10,10.5)\drawline(0,19.5)(10,19.5)\drawline(0,20.5)(10,20.5)
  \drawline(9.5,30)(9.5,20)\drawline(10.5,30)(10.5,20)\drawline(9.5,10)(9.5,0)\drawline(10.5,10)(10.5,0)
  \drawline(19.5,30)(19.5,20)\drawline(20.5,30)(20.5,20)\drawline(19.5,10)(19.5,0)\drawline(20.5,10)(20.5,0)
  \node()(5,15){s}\node()(15,35){$\temper=2$}\node()(5,25){1}\node()(5,5){2}
  \node()(15,25){3}\node()(15,5){4}\node()(25,25){5}\node()(25,5){6}\node()(25,15){7}
  \end{picture}
  \caption{Two counter-examples that square tile assembly system cannot be simulated by triangular
  tile assembly system. Each glue is unique and thus the label is omitted.}\label{fig:countereg}
\end{figure}
\end{proof}

The right super-tile in Figure~\ref{fig:countereg} is of shape with
a missing tile in the middle, and we call it has \emph{``hole''}.
More formally, we say a super-tile has no hole if it is full and for
every closed path of tiles, all enclosed region is occupied by
tiles.

\begin{lemma}
For any square tile assembly system $S$ under temperature
$\temper=1$ or under temperature $\temper=2$ with no hole in the
produce, there is a division of $S$ that can produce the same shape
with $\pi/4$ rotation.
\end{lemma}
\begin{proof}
For $\temper=1$, the proof is straightforward. For any square tile
$s_i$ with glues $\gamma_1,\gamma_2,\gamma_3,\gamma_4$ (on east,
north, west, south sides, respectively) in $S$, we replace it with a
pair of right triangular tiles $(i,\gamma_1,\gamma_2,{\tt n})$ and
$(i,\gamma_3,\gamma_4,{\tt s})$. Then the new right triangular tiles
is a division of $S$ and produce the same shape with $\pi/4$
rotation.

For $\temper=2$, now we assume there is no hole in the produce of
$S$.

First we prove that there is an assembly process
$st_0,st_1,st_2,\ldots,st_n$ such that every super-tile $st_i$ in
the process has no hole in it. Otherwise, we pick such a process
that the steps of the first appearance of a hole super-tile is the
largest among all assembly process. Let $st_i$ be the first
appearance of a hole super-tile. Then there is a tile $t$ in the
hole region that will stick to the super-tile later and there are
two adjacent edges that can stick to that tile due to the fullness
of the produce. So, we can add $t$ to the super-tile $st_{i-1}$ and
get a new $st_i'$ that does not have hole, which contradiction to
the choice of the process.

Now we prove that for the assembly process without hole super-tile,
the new tile can stick to the super-tile at each step by two
adjacent edges. Otherwise, suppose $st_{i-1}$ becomes $st_i$ by
sticking $t$ and $t$ only stick to $st_{i-1}$ either by north and
south sides or by east and west sides. Without loss of generality,
suppose it is by north and south sides. Since the produce is full,
there is no tiles on the east and on the west sides, or $t$ can
stick by two adjacent edges. But in this case, since $st_{i-1}$ is
connect, $st_i$ must contains a hole, which contradicts the fact
$st_i$ has no hole.

Since new tile can stick to the super-tile at each step by two
adjacent edges, we can add a pair of right triangular tiles to
simulate that square tiles; and we let the strength on the cutting
edges be $\geq\temper$. Do so for the whole assembly process, and we
get a new right triangular tiles, which is a division of $S$ and
produce the same shape with $\pi/4$ rotation.
\end{proof}

By the previous two lemmas, we see that square tile assembly systems
can be simulated by their division only under certain conditions.
Not we discuss the relation between the two types of tile assembly
systems in another direction: whether every right triangular tile
assembly system can be simulated by a square triangular tile
assembly system, where we assume the produce of the system is
compatible with the tiles.

\begin{lemma}\label{lem:counteregg}
There exists a right triangular tile assembly system $T$ such that
there is no square tile assembly system $S$ that produce the same
shape with $\pi/4$ rotation.
\end{lemma}
\begin{proof}
An example is illustrated in Figure~\ref{fig:counteregg}, where the
strength is illustrated by the number of parallel edges. The system
is under temperature $\temper=2$ and the top-left tile is the seed.
Let $S$ be any square tile assembly system that generate a
super-tile as in the example. Then $S$ will continue growing by
sticking a tile on the left bottom corner to the right top corner.
Therefore, the super-tile in the example cannot be produced by
square tile assembly system. \qedhere

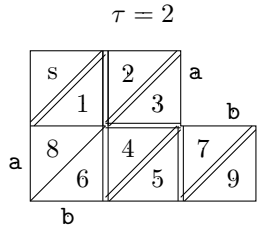
\begin{figure}
\centering
  \begin{picture}(30,20)
  \gasset{AHnb=0, Nw=0, Nh=0, Nframe=no}
  \drawpolygon(0,0)(30,0)(30,10)(20,10)(20,20)(0,20)
  \drawline(0,10.5)(9.5,20)\drawline(0.5,10)(10,19.5)\drawline(0,10)(10,10)
  \drawline(9.65,20)(9.65,0)\drawline(10.35,20)(10.35,0)
  \drawline(10,10.5)(19.5,20)\drawline(10.5,10)(20,19.5)
  \drawline(10,10.35)(20,10.35)\drawline(10,9.75)(20,9.75)
  \drawline(10,0.5)(19.5,10)\drawline(10.5,0)(20,9.5)
  \drawline(9.65,20)(9.65,10)\drawline(10.35,20)(10.35,10)
  \drawline(0,0)(10,10)\drawline(19.65,10)(19.65,0)\drawline(20.35,10)(20.35,0)
  \drawline(20,0.5)(29.5,10)\drawline(20.5,0)(30,9.5)
  \node()(3,17){s}\node()(15,25){$\temper=2$}\node()(7,13){1}\node()(13,17){2}\node()(17,13){3}
  \node()(13,7){4}\node()(17,3){5}\node()(7,3){6}\node()(23,7){7}\node()(3,7){8}\node()(27,3){9}
  \node()(-2,5){{\tt a}}\node()(5,-2){{\tt b}}
  \node()(22,17){{\tt a}}\node()(27,12){{\tt b}}
  \end{picture}
  \caption{A counter-examples that triangular tile assembly system cannot be simulated by
  square tile assembly system. Each glue, unless mentioned, is unique and thus the label is omitted.
  For convenience, the picture is rotated by $\pi/4$.}\label{fig:counteregg}
\end{figure}
\end{proof}

\begin{lemma}
For any right triangular tile assembly system $T$ under temperature
$\temper=1$, there is a square tile assembly system $S$ that can
produce the same shape with $\pi/4$ rotation.
\end{lemma}
\begin{proof}
For $\temper=1$, we construct a square tile assembly system $S$ from
$T$ as follows: for every pair of right triangular tiles
$(\gamma_1,\gamma_2,\gamma_3,{\tt
n}),(\gamma_1,\gamma_4,\gamma_5,{\tt s})$ or
$(\gamma_1,\gamma_2,\gamma_3,{\tt
e}),(\gamma_1,\gamma_4,\gamma_5,{\tt w})$ in $T$, where
$\gamma_1\neq\phi$, we add a new square tile with glues
$\gamma_4,\gamma_5,\gamma_2,\gamma_3$ or
$\gamma_5,\gamma_2,\gamma_3,\gamma_4$ (on east, north, west, south
sides, respectively) to $S$. Then the new square tile assembly
system produce the same shape with $\pi/4$ rotation.
\end{proof}

To compare the produces of two tile assembly system, we not only
compare the shape of the produce, but also compare the glues on
shared common edges, including both those on border and those inside
the super-tiles, with possible affine transformation on the shape,
which includes rotation, scaling, shift, and their compositions. We
call the power of produce certain super-tiles the \emph{shape
complexity} and say one system has greater power than another system
if every system produce of the former type with compatible shape is
the produce of some system of the latter type.

For every equilateral triangular tile assembly system $T$, there is
a right triangular tile assembly system $\mathcal{F}(T)$ such that
the produces of two system is equivalent up to an affine
transformation. There are two more types of tiles in right
triangular tile assembly system, which cannot be simulated by
equilateral triangular tiles. So we can say the shape complexity of
equilateral triangular tile assembly system is less than that of
right triangular tile assembly system.

The example given in Lemma~\ref{lem:counteregg} is a flattened
equilateral triangular tile assembly system. In other words, there
exists equilateral triangular tile assembly system which cannot be
produced by square tile assembly system even under affine
transformations. By Lemma~\ref{lem:countereg} and
Lemma~\ref{lem:counteregg}, we have the follow theorem.

\begin{theorem}
The square tile assembly systems and the triangular tile assembly
systems are not comparable in the sense of shape complexity.
\end{theorem}

\section{Assembly of Triangles}
Without loss of generality, we consider the assembly of an upright
full triangle. For downward full triangle, one can simply define a
new triangular tile system by flip-over each tiles in the original
triangular tile system.

\begin{proposition}
For temperature $\temper=1$ the minimal number of tiles to assemble
a full triangle with shortest edge of length $N$ is $N^2$, including
$N(N+1)/2$ upright triangular tiles and $N(N-1)/2$ downright
triangular tiles.
\end{proposition}
\begin{proof}
Without loss of generality, we consider the assembly of equilateral
triangles by equilateral triangular tiles. For the case of right
triangles, we can treat it as a flattened equilateral triangles.

By Proposition~\ref{prop:t1}, there exists a system of $N^2$ tiles
to assemble the required full triangle. To show that it is optimal,
suppose there is a system $(T,s,g,1)$ with less tiles. Then by the
pigeon hole principle, there are two tiles $t_1,t_2$ in the produce
that are the same. Since the produce is a full triangle, there is a
non-crossing path of tiles from $s$ to $t_1$ and from $t_1$ to
$t_2$, respectively. Since the temperature $\temper=1$, there is a
possible assembly process that starts from $s$ and sticks each tile
along the path from $s$ to $t_1$ and then from $t_1$ to $t_2$. After
$t_2$ sticks to the super-tile, again sticks each tile along the
path from $t_1$ to $t_2$ ($t_1$ itself is not included). So the
produce of the system is a infinite structure. Since the system is
deterministic, the produce cannot be a triangle, which contradicts
the assumption. So the system with $N^2$ tiles is optimal.
\end{proof}

Now we consider the temperature $\temper\geq2$. First we show how to
use $2N-1$ triangular tiles, including $N+1$ upright tiles and $N$
downright tiles assemble a full triangular.

\begin{proposition}
For temperature $\temper=2$ there is a triangular tile assembly
system of $2N-1$ tiles that produces a full triangular with shortest
edge of length $N$.
\end{proposition}
\begin{proof}
The system is illustrated in Figure~\ref{fig:linear}. The
construction here works for both equilateral triangular tiles and
right triangular tiles. \qedhere

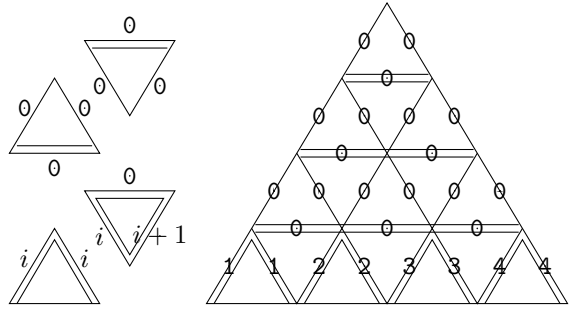
\begin{figure}
\centering
  \begin{picture}(25,35)
  \gasset{AHnb=0, Nw=0, Nh=0, Nframe=no}
  \drawpolygon(0,0)(12,0)(6,10)\drawline(1,0)(6,8.5)(11,0)\node()(2,6){$i$}\node()(10,6){$i$}
  \drawpolygon(22,15)(10,15)(16,5)\drawpolygon(20.5,14)(11.5,14)(16,6.5)\node()(12,9){$i$}\node()(20,9){$i+1$}\node()(16,17){$\tt0$}
  \drawpolygon(0,20)(12,20)(6,30)\drawline(1,21)(11,21)\node()(2,26){$\tt0$}\node()(10,26){$\tt0$}\node()(6,18){$\tt0$}
  \drawpolygon(22,35)(10,35)(16,25)\drawline(11,34)(21,34)\node()(12,29){$\tt0$}\node()(20,29){$\tt0$}\node()(16,37){$\tt0$}
  \end{picture}
  \begin{picture}(50,40)
  \gasset{AHnb=0, Nw=0, Nh=0, Nframe=no}
  \drawpolygon(0,0)(48,0)(24,40)
  \drawline(12,20)(18,10)(30,30)
  \drawline(18,30)(30,10)(36,20)
  \drawline(18.5,29.5)(29.5,29.5)\drawline(18.5,30.5)(29.5,30.5)
  \drawline(12.5,19,5)(35.5,19.5)\drawline(12.5,20.5)(35.5,20.5)
  \drawline(6.5,9.5)(41.5,9.5)\drawline(6.5,10.5)(41.5,10.5)
  \drawline(1,0)(6,8.5)(11,0)\drawline(13,0)(18,8.5)(23,0)
  \drawline(25,0)(30,8.5)(35,0)\drawline(37,0)(42,8.5)(47,0)
  \drawline(6,10)(12,0)(18,10)(24,0)(30,10)(36,0)(42,10)
  \node()(3,5){{\tt1}}\node()(9,5){{\tt1}}
  \node()(15,5){{\tt2}}\node()(21,5){{\tt2}}
  \node()(27,5){{\tt3}}\node()(33,5){{\tt3}}
  \node()(39,5){{\tt4}}\node()(45,5){{\tt4}}
  \node()(12,10){{\tt0}}\node()(24,10){{\tt0}}\node()(36,10){{\tt0}}
  \node()(9,15){{\tt0}}\node()(15,15){{\tt0}}\node()(21,15){{\tt0}}\node()(27,15){{\tt0}}\node()(33,15){{\tt0}}\node()(39,15){{\tt0}}
  \node()(18,20){{\tt0}}\node()(30,20){{\tt0}}
  \node()(15,25){{\tt0}}\node()(21,25){{\tt0}}\node()(27,25){{\tt0}}\node()(33,25){{\tt0}}
  \node()(24,30){{\tt0}}
  \node()(21,35){{\tt0}}\node()(27,35){{\tt0}}
  \end{picture}
  \caption{A triangular tile assembly system of $2N-1$ tiles produces a full triangle,
  where ${\tt1}\leq i\leq N$. On the right is an example for
  $N=4$}\label{fig:linear}
\end{figure}
\end{proof}

Using the same technique of square tile assembly for $N\times N$
squares~\cite{Rothemund&Winfree2000}, the following result follows.

\begin{proposition}
There is a right triangular tile assembly system of $O(\log N)$
tiles that produces a full right triangular with shortest edge of
length $N$.
\end{proposition}
\begin{proof}
The idea is that using a seed row of length $n=\lceil \log N\rceil$
to construct a $(n-1)\times(N-n+1)$ rectangle super-tile by counting
from $(2^{n}-N+n+2)/2$ to $2^{n-1}$ with duplicate copies. Then the
rectangle is completed by filling tiles to make a right triangle.
The temperature is $\temper=2$ and tiles is illustrated in
Figure~\ref{fig:log}, where $(s1,\phi,{\tt l},{\tt w})$ is the seed.
There are in total $2n+37$ tiles. \qedhere

\begin{figure}
\centering
  \rtilewdss{s1}{}{{\tt g}}\rtileedds{s1}{s2}{{\tt[1}}
  \rtilewddd{s2}{s2}{{\tt l}}\rtileedds{s2}{s3}{{\tt 0}}
  \rtilensdd{}{{\tt 0]i}}{s3}
  \rtilewdss{{\tt g}}{{\tt g}}{{\tt g}}\rtileedss{{\tt g}}{{\tt g}}{{\tt g}}

  \rtilesddd{{\tt 1]}}{{\tt 0]i}}{{\tt r}}
  \rtilesdds{{\tt 1]}}{{\tt 0]}}{{\tt g}}\rtilendss{{\tt 1]}}{{\tt 1]}}{{\tt n}}
  \rtilesdds{{\tt 0]}}{{\tt 1]}}{{\tt g}}\rtilendss{{\tt 0]}}{{\tt 0]}}{{\tt c}}

  \rtilesdss{{a}}{{a}}{{\tt n}}\rtilendss{{a}}{{a}}{{\tt n}}
  \rtilesdss{{\tt 1}}{{\tt 0}}{{\tt c}}\rtilendss{{\tt 1}}{{\tt 1}}{{\tt n}}
  \rtilesdss{{\tt 0'}}{{\tt 1}}{{\tt c}}\rtilendss{{\tt 0'}}{{\tt 0}}{{\tt c}}

  \rtilesdss{{\tt [}a}{{\tt [}a}{{\tt n}}\rtilendds{{\tt [}a}{{\tt [}a}{}
  \rtilesdss{{\tt [1}}{{\tt [0}}{{\tt c}}
  \rtilesdss{{\tt f[1}}{{\tt [1}}{{\tt c}}\rtilendss{{\tt f[1}}{}{}

  \rtilewdsd{{\tt [}a}{}{{\tt [}a}\rtileedss{{\tt [}a}{{\tt x}}{{\tt [}a}
  \rtilewdss{{a}}{{\tt x}}{{a}}\rtileedss{{a}}{{\tt x}}{{a}}
  \rtilewdss{{\tt a]}}{{\tt x}}{{\tt a]}}\rtileedsd{{\tt a]}}{{\tt g}}{{\tt a]}}

  \rtilensds{}{{\tt l}}{{\tt l}}
  \rtileedss{{\tt l}}{{\tt l}}{{\tt g}}
  \rtilewdsd{{\tt l}}{{\tt g}}{{\tt l}}
  \rtilenssd{}{{\tt r}}{{\tt r}}
  \rtilewdss{{\tt r}}{{\tt g}}{{\tt r}}
  \rtileedds{{\tt r}}{{\tt r}}{{\tt g}}

  \begin{picture}(100,50)
  \gasset{AHnb=0, Nw=0, Nh=0, Nframe=no}
  \drawpolygon(0,0)(100,0)(50,50)
  \drawline(10,0)(55,45)\drawline(20,0)(60,40)\drawline(35,5)(65,35)\drawline(40,0)(70,30)\drawline(50,0)(75,25)
  \drawline(60,0)(80,20)\drawline(70,0)(85,15)\drawline(80,0)(90,10)\drawline(90,0)(95,5)
  \drawline(90,0)(45,45)\drawline(80,0)(40,40)\drawline(70,0)(35,35)\drawline(60,0)(30,30)\drawline(50,0)(25,25)
  \drawline(40,0)(20,20)\drawline(25,5)(15,15)\drawline(15,5)(10,10)
  \node()(12.5,10){s}\node()(17.5,12.5){{\tt[1}}\node()(22.5,7.5){{\tt0}}\node()(27.5,2.5){{\tt0]i}}
  \drawline(9.5,0)(5,4.5)\drawline(10.5,0)(5.5,5)
  \drawline(9.65,9.65)(9.65,0)\drawline(10.35,9.65)(10.35,0)
  \drawline(19.5,0)(15,4.5)\drawline(20.5,0)(15.5,5)
  \drawline(24.5,5)(29.5,0)(34.5,5)\drawline(25.5,5)(30.5,0)(35.5,5)
  \drawline(34.5,14.5)(39.5,9.5)(39.5,0)\drawline(40,10)(40,0)
  \drawline(44.65,15)(44.65,5)\drawline(45.35,15)(45.35,5)
  \drawline(44.5,24.5)(49.5,19.5)(49.5,0)\drawline(50,20)(50,0)
  \drawline(54.65,25)(54.65,5)\drawline(55.35,25)(55.35,5)
  \drawline(54.5,34.5)(59.5,29.5)(59.5,0)\drawline(60,30)(60,0)
  \drawline(64.65,35)(64.65,5)\drawline(65.35,35)(65.35,5)
  \drawline(69.65,30)(69.65,0)\drawline(70.35,30)(70.35,0)
  \drawline(74.65,25)(74.65,5)\drawline(75.35,25)(75.35,5)
  \drawline(79.65,20)(79.65,0)\drawline(80.35,20)(80.35,0)
  \drawline(84.65,15)(84.65,5)\drawline(85.35,15)(85.35,5)
  \drawline(89.65,10)(89.65,0)\drawline(90.35,10)(90.35,0)
  \drawline(15,5)(15,15)(25,15)(25,25)(35,25)(35,35)(45,35)(45,45)(55,45)
  \drawline(14.3,5.3)(14.3,14.7)\drawline(15.7,15.7)(24.3,15.7)(24.3,24.3)
  \drawline(25.7,25.7)(34.3,25.7)(34.3,34.3)\drawline(35.7,35.7)(44.3,35.7)(44.3,44.3)
  \drawline(45.7,45.7)(54.3,45.7)
  \drawline(20,0)(20,10)(30,10)(30,20)(40,20)(40,30)(50,30)(50,40)(60,40)
  \drawline(19.3,0.3)(19.3,9.7)\drawline(20.7,10.7)(29.3,10.7)(29.3,19.3)
  \drawline(30.7,20.7)(39.3,20.7)(39.3,29.3)\drawline(40.7,30.7)(49.3,30.7)(49.3,39.3)
  \drawline(50.7,40.7)(59.3,40.7)
  \drawline(25,5)(35,5)(35,15)(45,15)(45,25)(55,25)(55,35)(65,35)
  \drawline(25.7,5.7)(34.3,5.7)(34.3,14.3)\drawline(35.7,15.7)(44.3,15.7)(44.3,24.3)
  \drawline(45.7,25.7)(54.3,25.7)(54.3,34.3)\drawline(55.7,35.7)(64.3,35.7)
  \drawline(19.5,19.5)(24.5,14.5)\drawline(29.5,29.5)(34.5,24.5)\drawline(39.5,39.5)(44.5,34.5)
  \drawline(20,9)(15,4)\drawline(24.5,5.5)(19.5,0.5)
  \drawline(41,0)(45.5,4.5)\drawline(51,0)(55.5,4.5)\drawline(61,0)(65.5,4.5)
  \drawline(71,0)(75.5,4.5)\drawline(81,0)(85.5,4.5)\drawline(91,0)(95.5,4.5)
  \end{picture}
  \caption{The $O(\log N)$ tiles, $a\in\set{{\tt 0,1}}$, and the produce for $N=10,n=4$.
  For simplicity, the label on the super-tile is omitted.}\label{fig:log}
\end{figure}
\end{proof}

Using the same technique of base conversion as appeared in the
square tile assembly~\cite{Adleman&Cheng&Goel&Huang2001} to count
the integer represented by tiles, the bound on the minimal number of
tiles required can be improved to $O(\log N/\log\log N)$, which is
optimal; the construction is under temperature $\temper=3$.

\begin{corollary}
There is a right triangular tile assembly system of $O(\log
N/\log\log N)$ tiles that produces a full right triangular with
shortest edge of length $N$.
\end{corollary}

\section{Conclusion}
Square tile assembly system is discussed in the literatures widely.
In this paper, we studied the triangular tile assembly system. We
showed that the triangular tile assembly system is also Turning
universal. The halting problem can be reduced to the tiling full
plane problem.

Compared to the square tile system, the triangular tile system need
more tiles to assemble a large compatible structure due to the fact
that a triangular tile has less edges than a square tile. In
general, as we showed, the two type of assembly system is not
comparable in the shape complexity. More precisely, there exists a
square tile assembly system $S$ such that no division of $S$
produces the same shape with $\pi/4$ rotation; and there exists a
triangular tile assembly system $T$ such that no square tile
assembly system produces the same shape, which is compatible with
square tiles, with same border glues with $\pi/4$ rotation.

We also discussed the assembly of triangles and the number of tiles
required to assemble a triangle with minimal edge of length $N$ is
$O(\log N/\log\log N)$, which is of the same order as those of
square tiles. The techniques used here is from that of assembly of
squares.

The model we used in this paper is of fixed temperature, unit
growing (at each step, only a single tile stick to the super-tile),
and irreversible (once tiles stick together, they will not break in
the further). There are other possible choice of models. For
example, if we allow variable temperature and reversible process as
discuss on square tiles~\cite{Adleman&Cheng&Goel&Huang2001}, then in
exactly the same way to the assembly of squares, one can prove
without difficulty that $O(1)$ tiles is enough to assemble arbitrary
large compatible triangles; in that case the time sequence is of
length $O(\log N)$.

All the result presented in the paper is based on theoretical study.
It will be interesting to assemble a triangle structure using
triangular tiles in the laboratory.

% \bibliographystyle{plain}
% \bibliography{triatile}

\end{document}